\documentclass[12pt]{article}

\usepackage{amsmath}
\usepackage{amssymb}
\usepackage{amsthm}
\usepackage[pdfpagemode=None]{hyperref}
\usepackage{graphicx}
\usepackage{pstricks}

\newtheorem{thm}{Theorem}[section]
\newtheorem{co}[thm]{Corollary}
\newtheorem{lem}[thm]{Lemma}

\newtheorem{assumption}[thm]{Assumption}

\newtheorem{pr}[thm]{Proposition}

\newtheorem{definition}[thm]{Definition}
\newenvironment{de}{\begin{definition}\rm}{\end{definition}}
\newtheorem{example}[thm]{Example}
\newenvironment{exmp}{\begin{example}\rm}{\end{example}}
\newtheorem{remark}[thm]{Remark}
\newenvironment{rem}{\begin{remark}\rm}{\end{remark}}
\newtheorem{tab}{Table}

\newcommand{\adj}{{\rm adj}\,}

\newcommand{\ord}{{\rm ord}\,}

\def\eps{\varepsilon}

\title{Asymptotics of Entropy Rate in Special Families of Hidden Markov Chains}

\author{\begin{tabular}{cc}
Guangyue Han&Brian Marcus\\
University of Hong Kong&University of British Columbia\\
{\em email:} ghan@maths.hku.hk&{\em email:} marcus@math.ubc.ca\\
\end{tabular}}

\date{{\normalsize \today}}

\begin{document}\maketitle\thispagestyle{empty}

\begin{abstract}
We generalize a result in~\cite{hm06b} and derive an asymptotic
formula for entropy rate of a hidden Markov chain around a ``weak
Black Hole''. We also discuss applications of the asymptotic formula
to the asymptotic behaviors of certain channels.
\end{abstract}


\section{Introduction}

Consider a discrete finite-valued stationary stochastic process
$Y=Y_{-\infty}^{\infty}:=\{Y_n: n \in \mathbb{Z}\}$. The entropy rate of $Y$ is defined to be
$$
H(Y)=\lim_{n \to \infty} H(Y_{-n}^0)/(n+1);
$$
here, $H(Y_{-n}^0)$ denotes the joint entropy of $Y_{-n}^0 := \{Y_{-n}, Y_{-n+1}, \cdots, Y_{0}\}$, and $\log$ is taken to mean the natural logarithm.

If $Y$ is a Markov chain with alphabet
$\{1, 2, \cdots, B\}$ and transition probability matrix $\Delta$, it
is well known that $H(Y)$ can be explicitly expressed with the
stationary vector of $Y$ and $\Delta$. A function
$Z=Z_{-\infty}^{\infty}$ of the Markov chain $Y$ with the form
$Z=\Phi(Y)$ is called a {\em hidden Markov chain}; here $\Phi$ is a
function defined on $\{1, 2, \cdots, B\}$, taking values in $\mathcal{A}:=\{1, 2, \cdots, A\}$ (alternatively a hidden
Markov chain is defined as a Markov chain observed in noise). For a
hidden Markov chain, $H(Z)$ turns out (see Equation~(\ref{blackwell_form}))
to be the integral of a certain function defined on a simplex with respect to a measure due
to Blackwell~\cite{bl57}. However Blackwell's measure is somewhat
complicated and the integral formula appears to be difficult to
evaluate in most cases. In general it is very difficult to compute $H(Z)$;
so far there is no simple and explicit formula for $H(Z)$.

Recently, the problem of computing the entropy rate of a hidden
Markov chain $Z$ has drawn much interest, and many approaches have been
adopted to tackle this problem. For instance, Blackwell's measure
has been used to bound the entropy rate~\cite{or03} and a variation
on the Birch bound~\cite{bi62} was introduced in~\cite{eg04}. An
efficient Monte Carlo method for computing the entropy rate of a
hidden Markov chain was proposed independently by Arnold and
Loeliger~\cite{ar01}, Pfister et. al.~\cite{pf01}, and Sharma and
Singh~\cite{sh01}. The connection between the entropy rate of a
hidden Markov chain and the top Lyapunov exponent of a random matrix
product has been observed~\cite{ho03, ho06,jss04,gh95u}.
In~\cite{gm05}, it is shown that under mild positivity assumptions
the entropy rate of a hidden Markov chain varies analytically as a
function of the underlying Markov chain parameters.

Another recent approach is based on computing the coefficients of an
asymptotic expansion of the entropy rate around certain values of
the Markov and channel parameters. The first result along these
lines was presented in \cite{jss04}, where for a binary symmetric
channel with crossover probability $\eps$ (denoted by BSC($\eps$)),
the Taylor expansion of $H(Z)$ around $\eps=0$ is studied for a
binary hidden Markov chain of order one. In particular, the first
derivative of $H(Z)$ at $\eps=0$ is expressed very compactly as a
Kullback-Liebler divergence between two distributions on binary
triplets, derived from the marginal of the input process $X$.
Further improvements and new methods for the asymptotic expansion
approach were obtained in~\cite{or04},~\cite{zu05},~\cite{zu06}
and~\cite{hm06b}. In~\cite{or04} the authors express the entropy
rate for a binary hidden Markov chain where one of the transition
probabilities is equal to zero as an asymptotic expansion including
a $O(\eps\log\eps)$ term.

This paper is organized as follows. In Section~\ref{II} we give an
asymptotic formula (Theorem~\ref{main}) for the entropy rate of a hidden Markov chain
around a weak Black Hole. The coefficients in the formula can be computed in principle (although
explicit computations may be quite complicated in general). The formula can be viewed as a generalization of the Black Hole condition considered in~\cite{hm06b}. The {\em weak} Black Hole case is important for hidden Markov chains
obtained as output processes of noisy channels, corresponding to input processes,
{\em for which certain sequences have probability zero}. Examples are given in Section~\ref{III}. Example~\ref{treated} was already treated in~\cite{hm06a} for only the first few coefficients; but in this case, these coefficients were computed quite explicitly.

\section{Asymptotic Formula for Entropy Rate}  \label{II}

Let $W$ be the simplex, comprising the vectors
$$
\{w=(w_1, w_2, \cdots, w_B) \in \mathbb{R}^B:w_i \geq 0, \sum_i
w_i=1\},
$$
and let $W_a$ be all $w \in W$ with $w_i=0$ for $\Phi(i) \neq a$.
For $a \in \mathcal{A}$, let $\Delta_a$ denote the $B \times B$ matrix such
that $\Delta_a(i, j)=\Delta(i, j)$ for $j$ with $\Phi(j)=a$, and
$\Delta_a(i, j)=0$ otherwise. For $a \in A$, define the
scalar-valued and vector-valued functions $r_a$ and $f_a$ on $W$ by
$$
r_a(w)= w \Delta_a \mathbf{1},
$$
and
$$
f_a(w)=w \Delta_a/ r_a(w).
$$
Note that $f_a$ defines the action of the matrix $\Delta_a$ on the
simplex $W$.

If $Y$ is irreducible, it turns out that
\begin{equation} \label{blackwell_form}
H(Z)=-\int \sum_a r_a(w) \log r_a(w) dQ(w),
\end{equation}
where $Q$ is {\em Blackwell's measure}~\cite{bl57} on $W$. This
measure, which satisfies an integral equation dependent on the
parameters of the process, is however very hard to extract from the
equation in any explicit way.

\begin{de} (see~\cite{hm06b})
Suppose that for every $a \in \mathcal{A}$, $\Delta_a$ is a rank one matrix,
and every column of $\Delta_a$ is either strictly positive or all
zeros. We call this the {\em Black Hole} case.
\end{de}
It was shown~\cite{hm06b} that $H(Z)$ is analytic around a Black
Hole and the derivatives of $H(Z)$ can be exactly computed around a
Black Hole. In this sequel, we consider weakened assumptions and
prove an asymptotic formula for entropy rate of a hidden Markov
chain around a ``weak Black Hole'', generalizing the corresponding
result in~\cite{hm06b}.

\begin{de}
Suppose that for every $a \in \mathcal{A}$, $\Delta_a$ is either an all zero matrix or a
rank one matrix. We call this the {\em weak Black Hole} case.
\end{de}

We use the standard notation: by $\alpha = \Theta(\beta)$, we mean there exist positive constants $C_1, C_2$ such that $C_1 |\beta| \leq | \alpha | \leq C_2 |\beta|$, while by $ \alpha = O(\beta)$, we mean there exists a
positive constant $C$ such that $| \alpha | \leq C |\beta|$. For a given analytic function $f(\eps)$ around $\eps=0$, let $\ord(f(\eps))$ denote its order, i.e., the degree of the first non-zero term of its Taylor series expansion around $\eps=0$. Note that for an analytic function $f(\eps)$ around $\eps=0$,
$$
f(\eps)=\Theta(\eps^k) \Longleftrightarrow \ord(f(\eps))=k.
$$
We say $\Delta(\eps)$ is {\em normally parameterized} by $\eps$ ($\eps \geq 0$) if
\begin{enumerate}
\item each entry of $\Delta(\eps)$ is an analytic function at $\eps=0$,
\item when $\eps > 0$, $\Delta(\eps)$ is (non-negative and) irreducible,
\item $\Delta(0)$ is a weak black hole.
\end{enumerate}

In the following, expressions like $p_X(x)$ will be used to mean $P(X=x)$ and we drop the subscripts if the context is
clear: $p(x), p(z)$ mean $P(X=x), P(Z=z)$, respectively, and further $p(y|x), p(z_0|z_{-n}^{-1})$ mean $P(Y=y|X=x), P(Z_0=z_0|Z_{-n}^{-1}=z_{-n}^{-1})$, respectively.

\begin{pr}   \label{jp-analytic}
Suppose that $\Delta(\eps)$ is analytically parameterized by $\eps \geq 0$ and when $\eps > 0$, $\Delta(\eps)$ is non-negative and irreducible. Then for any fixed hidden Markov sequence $z_{-n}^{0} \in \mathcal{A}^{n+1}$,
\begin{enumerate}
\item $p(z_{-n}^{-1})$ is analytic around $\eps=0$;
\item $p(y_i=\cdot \;|z^i_{-n}):=(p(y_i=b|z^i_{-n}): b=1, 2, \cdots, B)$ is analytic around $\eps=0$, where
$\cdot$ denotes $B$ possible states of Markov chain $Y$,
\item $p(z_0|z_{-n}^{-1})$ is analytic around $\eps=0$.
\end{enumerate}
\end{pr}

\begin{proof}

{\it 1.} When $\eps > 0$, $\Delta(\eps)$ is non-negative and irreducible. By Perron-Frobenius theory~\cite{se80},
$\Delta(\eps)$ has a unique positive stationary vector, say $\pi(\eps)$. Since
$$
\adj(I-\Delta(\eps)) (I-\Delta(\eps))=\det(I-\Delta(\eps)) I=0
$$
(here $\adj(\cdot)$ denotes the adjugate operator on matrices), one
can choose $\pi(\eps)$ to be any normalized row vector of
$\adj(I-\Delta(\eps))$. So $\pi (\eps)$ can be written as
$$
\frac{(\pi_1 (\eps), \pi_2 (\eps), \cdots, \pi_B (\eps))}{\pi_1
(\eps)+\pi_2 (\eps)+\cdots+\pi_B (\eps)},
$$
where $\pi_i (\eps)$'s are non-negative analytic functions of $\eps$ and the first non-zero
term of every $\pi_i(\eps)$'s Taylor series expansion has a positive coefficient. Then we conclude that for each $i$
$$
\ord(\pi_i (\eps)) \geq \ord(\pi_1 (\eps)+\cdots+\pi_B (\eps)),
$$
and thus $\pi(\eps)$, which is uniquely defined on $\eps > 0$, can
be continuously extended to $\eps=0$ via setting $\pi(0)=\lim_{\eps
\to 0} \pi(\eps)$.

Now
\begin{equation}   \label{dadada}
p(z_{-n}^{-1})=\pi (\eps) \Delta_{z_{-n}} \cdots \Delta_{z_{-1}}
\mathbf{1}=\frac{(\pi_1 (\eps), \pi_2 (\eps), \cdots, \pi_B (\eps))
\Delta_{z_{-n}} \cdots \Delta_{z_{-1}} \mathbf{1}}{\pi_1
(\eps)+\pi_2 (\eps)+\cdots+\pi_B (\eps)}=:\frac{f(\eps)}{g(\eps)},
\end{equation}
here $\ord(f(\eps)) \geq \ord(g(\eps))$. It then follows that $p(z_{-n}^{-1})$
is analytic around $\eps=0$.

{\it 2.} Let $x_{i, -n}=x_{i, -n}(z^i_{-n})$ denote $p(y_i=\cdot \;|z^i_{-n})$. Then one
checks that $x_{i, -n}$ satisfies the following iteration:
\begin{equation}  \label{iteration}
x_{i, -n}=\frac{x_{i-1, -n} \Delta_{z_i}}{x_{i-1, -n} \Delta_{z_i}
\mathbf{1}},  \qquad -n \leq i \leq -1,
\end{equation}
starting with $x_{-n-1, -n} = p(y_{-n-1}=\cdot\;)$. Because $\Delta$ is analytically
parameterized by $\eps$ ($\eps \geq 0$) and $\Delta(\eps)$ is non-negative and irreducible when $\eps > 0$, inductively
we can prove (the proof is similar to the proof of {\it 1.}) that for any $i$, $x_{i, -n}$ can be
written as follows:
$$
x_{i, -n}=\frac{(f_1(\eps), f_2(\eps), \cdots, f_B(\eps))}{f_1(\eps)+f_2(\eps)+\cdots+f_B(\eps)},
$$
where $f_i(\eps)$'s are analytic functions around $\eps=0$. Note that for
each $i$
$$
\ord(f_i(\eps)) \geq \ord(f_1(\eps)+f_2(\eps)+\cdots+f_B(\eps)).
$$
The existence of the Taylor series expansion of $x_{i, -n}$ around
$\eps=0$ (for any $i$) then follows.

{\it 3.} One checks that
\begin{equation}   \label{negative1}
p(z_0|z_{-n}^{-1})=x_{-1, -n} \Delta_{z_0} \mathbf{1}.
\end{equation}
Analyticity of $p(z_0|z_{-n}^{-1})$ immediately follows from (\ref{negative1}) and analyticity of $x_{-1, -n}$ around $\eps=0$, which has been shown in {\it 2.}.

\end{proof}

\begin{lem}   \label{independence}
Consider two formal series expansion $f(x), g(x) \in
\mathbb{R}[[x]]$ such that $f(x)=\sum_{i=0}^{\infty} f_i x^i$ and
$g(x)=\sum_{i=0}^{\infty} g_i x^i$, where $g_0 \neq 0$. Let $h(x)
\in \mathbb{R}[[x]]$ be the quotient of $f(x)$ and $g(x)$ with
$h(x)=\sum_{i=0}^{\infty} h_i x^i$. Then $h_i$ is a function
dependent only on $f_0, \cdots, f_i$ and $g_0, \cdots, g_i$.
\end{lem}

\begin{proof}

Comparing the coefficients of all the terms in the following
identity:
$$
\left(\sum_{i=0}^{\infty} h_i x^i \right) \left(\sum_{i=0}^{\infty}
g_i x^i \right)=\sum_{i=0}^{\infty} f_i x^i,
$$
we obtain that for any $i$,
$$
h_0 g_i+h_1 g_{i-1}+\cdots+ h_i g_0=f_i.
$$
The lemma then follows from an induction (on $i$) argument.

\end{proof}

By Proposition~\ref{jp-analytic}, for any hidden Markov string
$z_{-m}^0$, the Taylor series expansion of $p(z_0|z_{-m}^{-1})$
around $\eps=0$ exists. We use $b_j(z_{-m}^0)$ to represent the
coefficient of $\eps^j$ in the expansion, namely
\begin{equation}    \label{expansion-1}
p(z_0|z_{-m}^{-1})=b_0(z_{-m}^0)+b_1(z_{-m}^0) \eps +b_2(z_{-m}^0)
\eps^2+\cdots.
\end{equation}

The following lemma shows that under certain conditions, some coefficients $b_j(z_{-m}^0)$ ``stabilize''. More precisely, we have:
\begin{lem}    \label{StabilizingLemma}
Consider a hidden Markov chain $Z$ with normally parameterized $\Delta(\eps)$. For two fixed hidden Markov chain sequences
$z_{-m}^0, \hat{z}_{-\hat{m}}^0$ such that
$$
z_{-n}^0=\hat{z}_{-n}^0, \qquad \ord(p(z_{-n}^{-1}|z_{-m}^{-n-1})), \;\; \ord(p(\hat{z}_{-n}^{-1}|\hat{z}_{-\hat{m}}^{-n-1})) \leq k
$$
for some $n \leq m, \hat{m}$ and some $k$, we have for $j$ with
$0 \leq j \leq n-4k-1$,
$$
b_j(z_{-m}^0)=b_j(\hat{z}_{-\hat{m}}^0).
$$
\end{lem}

\begin{proof}

Recall that $x_{i, -m}=x_{i, -m}(z^i_{-m})=p(y_i=\cdot
\;|z^i_{-m})$ and $\hat{x}_{i, -\hat{m}}=\hat{x}_{i, -\hat{m}}(\hat{z}^i_{-\hat{m}})=p(y_i=\cdot
\;|\hat{z}^i_{-\hat{m}})$, where $\cdot$ denotes the possible states of Markov
chain $Y$. Consider the Taylor series expansion of $x_{i, -m}, \hat{x}_{i, -\hat{m}}$
around $\eps=0$,
\begin{equation}   \label{TaylorSeries_1}
x_{i, -m}=a_0(z_{-m}^i)+a_1(z_{-m}^i) \eps +a_2(z_{-m}^i)
\eps^2+\cdots
\end{equation}
\begin{equation}   \label{TaylorSeries_2}
\hat{x}_{i, -\hat{m}}=a_0(z_{-\hat{m}}^i)+a_1(z_{-\hat{m}}^i) \eps +a_2(z_{-\hat{m}}^i)
\eps^2+\cdots
\end{equation}
We shall show that $a_j(z_{-m}^i)=a_j(\hat{z}_{-\hat{m}}^i)$ for $j$ with
$$
0 \leq j \leq n+i-\sum_{l=-n}^i \max \{J(z_{-m}^l), J(\hat{z}_{-\hat{m}}^l)\},
$$
where for any hidden Markov sequence $z_{-m}^i$,
$$
J(z_{-m}^i)=\left \{\begin{array}{cc}
1+\ord(p(z_i|z_{-m}^{i-1})) &\mbox{ if } \ord(p(z_i|z_{-m}^{i-1})) > 0 \\
\\
0 & \mbox{ if } \ord(p(z_i|z_{-m}^{i-1})) = 0 \\
\end{array}
\right. .
$$

Recall that
\begin{equation}  \label{m-iteration}
x_{i+1, -m}=\frac{x_{i, -m} \Delta_{z_{i+1}}(\eps)}{x_{i, -m}
\Delta_{z_{i+1}}(\eps) \mathbf{1}}.
\end{equation}
Now with (\ref{TaylorSeries_1}) and (\ref{TaylorSeries_2}), we have
\begin{equation} \label{moredetail}
x_{i, -m} \Delta_{z_{i+1}}(\eps)=\sum_{j=0}^{\infty} a_j(z_{-m}^i)
\sum_{k=0}^{\infty} \frac{\Delta^{(k)}_{z_{i+1}}(0)}{k!} \eps^k
=\sum_{l=0}^{\infty} c_l(z_{-m}^{i+1}) \eps^l,
\end{equation}
where superscript ${}^{(k)}$ denotes the $k$-th order derivative with respect to $\eps$.

We proceed by induction on $i$ (from $-n$ to $-1$).

First consider the case when
$i=-n$. When $\max \{J(z_{-m}^{-n}), J(\hat{z}_{-\hat{m}}^{-n})\} > 0$, the statement is vacuously true; when $J(z_{-m}^{-n})=J(\hat{z}_{-\hat{m}}^{-n})=0$, necessarily $\Delta_{z_{-n}}(0)$ is a rank one matrix, $a_0(z_{-m}^{-n-1}) \Delta_{z_{-n}}(0) \mathbf{1} > 0$ and $a_0(\hat{z}_{-\hat{m}}^{-n-1}) \Delta_{z_{-n}}(0) \mathbf{1} > 0$. Then we have
$$
a_0(z_{-m}^{-n})=\frac{a_0(z_{-m}^{-n-1}) \Delta_{z_{-n}}(0)}{a_0(z_{-m}^{-n-1}) \Delta_{z_{-n}}(0) \mathbf{1}} \stackrel{(*)}{=}\frac{a_0(-\hat{z}_{-\hat{m}}^{-n-1}) \Delta_{z_{-n}}(0)}{a_0(\hat{z}_{-\hat{m}}^{-n-1}) \Delta_{z_{-n}}(0) \mathbf{1}}=a_0(\hat{z}_{-\hat{m}}^{-n}),
$$
where $(*)$ follows from the fact that $\Delta_{z_{-n}}(0)$ is a rank one matrix.

Now suppose $i \geq -n$ and that $a_j(z_{-m}^i)=a_j(\hat{z}_{-\hat{m}}^i)$ for $j$ with $0 \leq j \leq n+i-\sum_{l=-n}^i \max \{J(z_{-m}^l), J(\hat{z}_{-\hat{m}}^l)\}$.

If $\ord(p(z_{i+1}|z_{-m}^i)) > 0$, since the leading coefficient vector of the Taylor series expansion in (\ref{moredetail}) is non-negative, $c_j(z_{-m}^{i+1}) \equiv \mathbf{0}$ for all $j$ with $0 \leq j \leq J(z_{-m}^{i+1})-2$ and $c_{J(z_{-m}^{i+1})-1}(z_{-m}^{i+1}) \not \equiv \mathbf{0}$. So applying Lemma~\ref{independence} to the following expression
\begin{equation}   \label{iteration}
x_{i+1, -m}=\frac{c_0(z_{-m}^{i+1})+c_1(z_{-m}^{i+1}) \eps +\cdots+c_l(z_{-m}^{i+1})
\eps^l+ \cdots}{c_0(z_{-m}^{i+1}) \mathbf{1}+c_1(z_{-m}^{i+1}) \mathbf{1} \eps
+\cdots+c_l(z_{-m}^{i+1}) \mathbf{1} \eps^l+ \cdots}=\frac{\sum_{l=0}^{\infty} c_{l+J(z_{-m}^{i+1})-1}(z_{-m}^{i+1})
\eps^l}{\sum_{l=0}^{\infty} c_{l+J(z_{-m}^{i+1})-1}(z_{-m}^{i+1}) \mathbf{1}
\eps^l},
\end{equation}
we conclude that for all $j$, $a_j(z_{-m}^{i+1})$ depends only on
$$
c_l(z_{-m}^{i+1}), \qquad J(z_{-m}^{i+1})-1 \leq l \leq J(z_{-m}^{i+1})-1+j,
$$
implying that $a_j(z_{-m}^{i+1})$ depends only on (or some of)
$$
a_l(z_{-m}^{i}), \qquad \Delta^{(l)}_{z_{i+1}}(0), \qquad 0 \leq l \leq J(z_{-m}^{i+1})-1+j.
$$
A completely parallel argument also applies to the case when $\ord(p(\hat{z}_{i+1}|\hat{z}_{-\hat{m}}^i)) > 0$. More specifically, the statements above for the case $\ord(p(z_{i+1}|z_{-m}^i)) > 0$ are still true if we replace $z, x, m$ with $\hat{z}, \hat{x}, \hat{m}$, which implies that $a_j(\hat{z}_{-\hat{m}}^{i+1})$ depends only on (or some of)
$$
a_l(\hat{z}_{-\hat{m}}^{i}), \qquad \Delta^{(l)}_{\hat{z}_{i+1}}(0), \qquad 0 \leq l \leq J(\hat{z}_{-\hat{m}}^{i+1})-1+j.
$$
Thus when $\max \{J(z_{-m}^{i+1}), J(\hat{z}_{-\hat{m}}^{i+1})\} > 0$, we have $a_j(z_{-m}^{i+1})=a_j(\hat{z}_{-\hat{m}}^{i+1})$ for $j$ with
$$
0 \leq j \leq n+i-\sum_{l=-n}^i \max \{J(z_{-m}^l), J(\hat{z}_{-\hat{m}}^l)\}-\max \{J(z_{-m}^{i+1})-1, J(\hat{z}_{-\hat{m}}^{i+1})-1\}
$$
$$
=n+(i+1)-\sum_{l=-n}^{i+1} \max \{J(z_{-m}^l), J(\hat{z}_{-\hat{m}}^l)\}.
$$

If $\ord(p(z_{i+1}|z_{-m}^i)) =0$, by (\ref{negative1}) necessarily we have
$$
a_0(z_{-m}^i) \Delta_{z_{i+1}}(0) \mathbf{1} \neq 0.
$$
Again by Lemma~\ref{independence} applied to expression
(\ref{iteration}), for any $j$, $a_j(z_{-m}^{i+1})$ depends only on
$$
a_l(z_{-m}^{i}), \qquad \Delta^{(l)}_{z_{i+1}}(0), \qquad 0 \leq l \leq j,
$$
Similarly if $\ord(p(\hat{z}_{i+1}|\hat{z}_{-m}^i)) =0$, we deduce that for any $j$, $a_j(\hat{z}_{-m}^{i+1})$
depends only on
$$
a_l(\hat{z}_{-m}^{i}), \qquad \Delta^{(l)}_{\hat{z}_{i+1}}(0), \qquad 0 \leq l \leq j.
$$
Thus if $\max \{J(z_{-m}^{i+1}), J(\hat{z}_{-\hat{m}}^{i+1})\}=0$, for any $j$ with
$$
0 \leq j \leq n+i-\sum_{l=-n}^i \max \{J(z_{-m}^l), J(\hat{z}_{-\hat{m}}^l)\}
= n+i-\sum_{l=-n}^{i+1} \max \{J(z_{-m}^l), J(\hat{z}_{-\hat{m}}^l)\},
$$
we have $a_j(z_{-m}^{i+1})=a_j(\hat{z}_{-m}^{i+1})$.

Now,
let $t = n+(i+1)-\sum_{l=-n}^{i+1} \max \{J(z_{-m}^l), J(\hat{z}_{-\hat{m}}^l)\}$.
Then one can show that
$$
a_t(z_{-m}^{i+1})=\frac{a_{t}(z_{-m}^i) \Delta_{z_{i+1}}(0)
a_0(z_{-m}^i) \Delta_{z_{i+1}}(0) \mathbf{1}-a_0(z_{-m}^i)
\Delta_{z_{i+1}}(0) a_{t}(z_{-m}^i) \Delta_{z_{i+1}}(0)
\mathbf{1}}{(a_0(z_{-m}^i) \Delta_{z_{i+1}}(0) \mathbf{1})^2}+\mbox{
other terms },
$$
where the first term in the expression above is equal to $\mathbf{0}$
(since $\Delta_{z_{i+1}}(0)$ is a rank one matrix), and
the ``other terms'' are functions of
\begin{equation}
\label{quantities}
a_0(z_{-m}^{i}), \cdots, a_{t-1}(z_{-m}^{i}), \Delta^{(0)}_{z_{i+1}}(0), \cdots,
\Delta^{(t)}_{z_{i+1}}(0).
\end{equation}
It follows that $a_j(z_{-m}^{i+1})$ is a function of the same quantities in
(\ref{quantities}).
%
By a completely parallel argument as above,
$a_j(\hat{z}_{-\hat{m}}^{i+1})$ is the same function of of the same quantities in
(\ref{quantities}).
So we have $a_j(z_{-m}^{i+1})=a_j(\hat{z}_{-\hat{m}}^{i+1})$ for $j$ with
$$
0 \leq j \leq n+(i+1)-\sum_{l=-n}^{i+1} \max \{J(z_{-m}^l), J(\hat{z}_{-\hat{m}}^l)\}.
$$

Notice that
$$
\sum_{l=-n}^{-1} \max \{J(z_{-m}^l), J(\hat{z}_{-\hat{m}}^l)\} \leq \sum_{l=-n}^{-1} (J(z_{-m}^l)+J(\hat{z}_{-\hat{m}}^l)) \leq 4k.
$$
The lemma then immediately follows from (\ref{negative1}) and the
proven fact that $a_j(z_{-m}^{-1})=a_j(\hat{z}_{-\hat{m}}^{-1})$ for $j$ with
$$
0 \leq j \leq n-1-\sum_{l=-n}^{-1} \max \{J(z_{-m}^l), J(\hat{z}_{-\hat{m}}^l)\}.
$$
\end{proof}

For a mapping $v = v(\eps): [0, \infty) \to W$ analytic at $\eps=0$ and a hidden Markov sequence $z_{-n}^0$, define
$$
p_v(z_{-n}^{-1}) = v \Delta_{z_{-n}} \cdots \Delta_{z_{-1}} {\bf{1}}, \mbox{ and } p_v(z_0|z_{-n}^{-1}) = \frac{p_v(z_{-n}^0) }{p_v(z_{-n}^{-1}) }.
$$
Let $b_{v, j}(z_{-n}^0)$ denote the coefficient of $\eps^j$ in the Taylor series expansion of $p_v(z_0|z_{-n}^{-1})$ (note that $b_{v, j}(z_{-n}^0)$ does not depend on $\eps$),
$$
p_v(z_0|z_{-n}^{-1}) = \sum_{j=0}^\infty b_{v, j}(z_{-n}^0) \eps^j.
$$
Using the same inductive approach in Lemma~\ref{StabilizingLemma}, we can prove that

\begin{lem}
For two mappings $v = v(\eps), \hat{v}=\hat{v}(\eps): [0, \infty) \to W$ analytic at $\eps=0$, if $\ord(p_v(z_{-n}^{-1})), \ord(p_{\hat{v}}(z_{-n}^{-1})) \le k$, we then have
$$
b_{v, j}(z_{-n}^0) =  b_{\hat{v}, j}(z_{-n}^0), ~~  0 \le j \le n - 4k -1.
$$
\end{lem}

Note that for $n \leq m, \hat{m}$, if $v(\eps)$ (or $\hat{v}(\eps)$) is equal to $p(y_{n-1}=\cdot|z_{-m}^{-n-1})$ (or $p(y_{n-1}=\cdot|z_{-\hat{m}}^{-n-1})$), then $p_v(z_{-n}^0)$ (or $p_{\hat{v}}(z_{-n}^0)$) will be equal to $p(z_{-n}^0|z_{-m}^{-n-1})$ (or $p(z_{-n}^0|z_{-\hat{m}}^{-n-1})$); and if for a Markov state $y$, $v(\eps)$ (or $\hat{v}(\eps)$) is equal to $p(y_{n-1}=\cdot|z_{-m}^{-n-1}y)$ (or $p(y_{n-1}=\cdot|z_{-\hat{m}}^{-n-1}y)$), then $p_v(z_{-n}^0)$ (or $p_{\hat{v}}(z_{-n}^0)$) will be equal to $p(z_{-n}^0|z_{-m}^{-n-1}y)$ (or $p(z_{-n}^0|z_{-\hat{m}}^{-n-1}y)$). It then immediately follows that
\begin{co}  \label{mixedwithy}
Given fixed sequences $z_{-m}^0, z_{-\hat{m}}^0, z_{-m}^0y_{-m-1}, \hat{z}_{-\hat{m}}^0y_{-\hat{m}-1}$ with $z_{-n}^0=\hat{z}_{-n}^0$
such that
$$
\ord(p(z_{-n}^{-1}|z_{-m}^{-n-1})), \ord(p(\hat{z}_{-n}^{-1}|\hat{z}_{-\hat{m}}^{-n-1})), \ord(p(z_{-n}^{-1}|z_{-m}^{-n-1}y_{-m-1})), \ord(p(\hat{z}_{-n}^{-1}|\hat{z}_{-\hat{m}}^{-n-1}y_{-\hat{m}-1})) \leq k,
$$
for $n \leq m, \hat{m}$ and some $k$, we have for $j$ with
$0 \leq j \leq n-4k-1$,
\begin{equation} \label{AllEqual}
b_j(z_{-m}^0y_{-m-1})=b_j(\hat{z}_{-\hat{m}}^0y_{-\hat{m}-1})=b_j(z_{-m}^0)=b_j(\hat{z}_{-\hat{m}}^0),
\end{equation}
where slightly abusing the notation, we define $b_j(z_{-m}^0y_{-m-1})$, $b_j(\hat{z}_{-\hat{m}}^0y_{-\hat{m}-1})$ as
the coefficients of the Taylor series expansions of $p(z_0|z_{-m}^0y_{-m-1})$, $p(\hat{z}_0|\hat{z}_{-\hat{m}}^0y_{-\hat{m}-1})$, respectively.
\end{co}

Consider expression (\ref{expansion-1}). In the following, we use $p^{<l>}(z_0|z_{-n}^{-1})$ to denote the truncated (up to the $(l+1)$-st term) Taylor series expansion of $p(z_0|z_{-n}^{-1})$, i.e.,
$$
p^{<l>}(z_0|z_{-n}^{-1})=b_0(z_{-n}^0)+b_1(z_{-n}^0) \eps
+b_2(z_{-n}^0) \eps^2+\cdots+b_l(z_{-n}^0) \eps^l.
$$

\begin{thm}   \label{main}
For a hidden Markov chain $Z$ with normally parameterized $\Delta(\eps)$, we have for any $k \geq 0$,
\begin{equation}   \label{mainFormula}
H(Z)=H(Z)|_{\eps=0}+\sum_{j=1}^{k+1} f_j \eps^j \log \eps+\sum_{j=1}^{k} g_j \eps^j + O(\eps^{k+1}),
\end{equation}
where $f_j$'s and $g_j$'s for $j=1, 2, \cdots, k+1$ are functions (more specifically, elementary functions built from log and polynomials) of $\Delta^{(i)}(0)$ for $0 \leq i \leq 6k+6$ and can be computed from $H_{6k+6}(Z(\eps))$.
\end{thm}

\begin{proof}

First fix $n$ such that $n \geq n_0=6k+6$. Consider the Birch upper bound on $H(Z)$
$$
H_n(Z):=H(Z_0|Z_{-n}^{-1})=- \sum_{z_{-n}^0} p(z_{-n}^0) \log
p(z_0|z_{-n}^{-1}).
$$
Note that for $j \geq k+2$,
\begin{equation}  \label{eq1}
\left| \sum_{\ord(p(z_{-n}^{0})) = j} p(z_{-n}^0) \log
p(z_0|z_{-n}^{-1}) \right| = O(\eps^{k+1}).
\end{equation}
So, in the following we only consider the sequences $z_{-n}^0$ with
$\ord(p(z_{-n}^0)) \leq k+1$. For such sequences,
since $\ord(p(z_0|z_{-n}^{-1})) \leq \ord(p(z_{-n}^0)) \leq k+1$, we have
\begin{equation}  \label{eq2}
|\log p(z_0|z_{-n}^{-1})-\log p^{<2k+1>}(z_0|z_{-n}^{-1})| =
O(\eps^{k+1});
\end{equation}
and by Lemma~\ref{StabilizingLemma}, we have
\begin{equation}  \label{eq3}
p^{<2k+1>}(z_0|z_{-n}^{-1})=p^{<2k+1>}(z_0|z_{-n_0}^{-1}).
\end{equation}

Now for any fixed $n \geq n_0$,
\begin{eqnarray}
\notag H_n(Z)&=&\sum_{z_{-n}^0}
-p(z_{-n}^0) \log p(z_0|z_{-n}^{-1})\\
\notag &\stackrel{(a)}{=}&\sum_{\ord(p(z_{-n}^{0})) \leq k+1}
-p(z_{-n}^0) \log p(z_0|z_{-n}^{-1})+O(\eps^{k+1})\\
\notag &\stackrel{(b)}{=}&\sum_{\ord(p(z_{-n}^{0})) \leq k+1}
-p(z_{-n}^0) \log p^{<2k+1>}(z_0|z_{-n}^{-1})+O(\eps^{k+1})\\
\notag &\stackrel{(c)}{=}&\sum_{\ord(p(z_{-n_0}^{0})) \leq k+1}
-p(z_{-n}^0) \log p^{<2k+1>}(z_0|z_{-n_0}^{-1})+O(\eps^{k+1})\\
\label{HNZ}  &=&\sum_{\ord(p(z_{-n_0}^{0})) \leq k+1}
-p(z_{-n_0}^0) \log p^{<2k+1>}(z_0|z_{-n_0}^{-1})+O(\eps^{k+1}),
\end{eqnarray}

where (a) follows from (\ref{eq1}); (b) follows from (\ref{eq2}); (c) follows from (\ref{eq3}), (\ref{eq1}) and the fact that
$$
\{z_{-n}^0: \ord(p(z_{-n_0}^{0})) \leq k+1\}
$$
$$=\{z_{-n}^0: \ord(p(z_{-n}^{0})) \leq k+1\} \cup \{z_{-n}^0: \ord(p(z_{-n_0}^{0})) \leq k+1, \ord(p(z_{-n}^{0})) \geq k+2\}.
$$
Expanding (\ref{HNZ}), we obtain:
$$
H_n(Z)=H(Z)|_{\eps=0}+\sum_{j=1}^{k+1} f_j \eps^j \log \eps+\sum_{j=1}^{k} g_j \eps^j + O(\eps^{k+1}),
$$
where $f_j$'s and $g_j$'s for $j=1, 2, \cdots, k+1$ are functions dependent only on $\Delta^{(i)}(0)$ for $0 \leq i \leq n_0$ and can be computed from $H_{n_0}(Z)$ (in fact for fixed $j$, $f_j$ and $g_j$ are functions dependent only on $\Delta^{(i)}(0)$ for $0 \leq i \leq 6j+6$ and can be computed from $H_{6j+6}(Z)$). In
particular,
\begin{equation} \label{eps}
\sum_{\ord(p(z_{-n}^0)) \leq k+1} \sum_{\ord(p(z_0|z_{-n_0}^{-1})) =
0} -p(z_{-n_0}^0) \log p^{<2k+1>}(z_0|z_{-n_0}^{-1})
\end{equation}
will contribute to $H(Z)|_{\eps=0}$ and the terms $\eps^j$, and
\begin{equation} \label{epslogeps}
\sum_{\ord(p(z_{-n}^0)) \leq k+1}
\sum_{\ord(p(z_0|z_{-n_0}^{-1})) > 0} -p(z_{-n_0}^0) \log
p^{<2k+1>}(z_0|z_{-n_0}^{-1})
\end{equation}
will contribute to the terms $\eps^j \log \eps$ and the terms
$\eps^j$.

Using Corollary~\ref{mixedwithy}, one can apply similar argument as above to the Birch lower bound
$$
\tilde{H}_n(Z):=H(Z_0|Z_{-n}^{-1} Y_{-n-1})=\sum_{z_{-n}^0,
y_{-n-1}} -p(z_{-n}^0 y_{-n-1}) \log p(z_0|z_{-n}^{-1} y_{-n-1}).
$$
For the same $n_0$, one can show that $\tilde{H}_n(Z)$ takes the
same form (\ref{HNZ}) as $H_n(Z)$, which implies that
$H_n(Z)$ and $\tilde{H}_n(Z)$ have exactly the same coefficients of
$\eps^j$ for $j \leq k$ and of $\eps^j \log \eps$ for $j \leq k+1$
when $n \geq n_0$. We thus prove the theorem.
\end{proof}

\begin{rem}
Theorem~\ref{main} still holds if we assume each entry of $\Delta(\eps)$ is merely a $C^{k+1}$ function of $\eps$ in a neighborhood of $\eps=0$: the proof still works if ``analytic'' is replaced by ``$C^{k+1}$'', and the Taylor series expansions are replaced by Taylor polynomials with remainder. We assumed analyticity of the parametrization only for simplicity.
\end{rem}

\begin{rem} Note that at a Black Hole, we have $\ord(p(z_0|z_{-n}^{-1})) = 0$ for any hidden Markov symbol
sequence $z_{-n}^0$.  Thus, from the discussion surrounding
expressions (\ref{eps}) and (\ref{epslogeps}) above, we see that
$f_j=0$ for all $j$. By the proof of Theorem~\ref{main}, Formula
(\ref{mainFormula}) is a Taylor polynomial with remainder; this is
consistent with the Taylor series formula for a Black Hole
in~\cite{hm06b}.
\end{rem}

\begin{rem}
The proof of Theorem~\ref{main} shows that for $n \geq n_0$,
$H_n(Z), \tilde{H}_n(Z)$ take the same form as in
(\ref{mainFormula}) with the same coefficients.
\end{rem}

\section{Applications to Finite-State Memoryless Channels at High Signal-to-Noise Ratio} \label{III}

Consider a finite-state memoryless channel with stationary input process.
Here, $C=\{C_n\}$ is an i.i.d. channel state process over finite alphabet $\mathcal{C}$ with $p_C(c)=q_c$ for $c \in  \mathcal{C}$, $X=\{ X_n \}$ is a stationary input process, independent of $C$, over finite alphabet $\mathcal{X}$
and $Z=\{Z_n \}$  is the resulting (stationary) output process over finite alphabet
$\mathcal{Z}$.  Let $p(z_n|x_n, c_n)=P(Z_n=z_n|X_n=x_n, C_n=c_n)$ denote the probability
that at time $n$, the channel output symbol is $z_n$ given that the
channel state is $c_n$ and the channel input is $x_n$. The mutual
information for such a channel is:
$$
I(X, Z):=H(Z)-H(Z|X)\stackrel{(*)}{=}H(Z)-\sum_{x \in \mathcal{X}, z \in \mathcal{Z}}
p(x, z) \log p(z|x),
$$
where $(*)$ follows from the memoryless property of the channel, and for $x \in \mathcal{X}, z \in \mathcal{Z}$,
$$
p(x, z)=\sum_{c \in \mathcal{C}} p(z|x, c)p(x)p(c), \qquad p(z|x)=\sum_{c \in \mathcal{C}} p(z|x, c) p(c).
$$

Now we introduce an alternative framework, using the concept of channel noise.  As above,
let $C$ be an i.i.d. channel state process, and let $X$ be a stationary input process, independent of $C$, over finite alphabets
 $\mathcal{C}$, $\mathcal{X}$. Let $\mathcal{E}$ (resp., $\mathcal{Z}$) be finite alphabets of abstract error events
(resp. output symbols) and let $\Phi: \mathcal{X} \times \mathcal{C} \times
\mathcal{E} \rightarrow \mathcal{Z}$ be a function.  For each $x \in \mathcal{X}$
and $c \in \mathcal{C}$, let $p(\cdot | x, c)$ be a conditional probability distribution
on $\mathcal{E}$.  This defines a jointly distributed stationary process $(X,C,E)$ over
$\mathcal{X} \times \mathcal{C} \times \mathcal{E}$. If $X$ is a first order Markov chain
with transition probability matrix $\Pi$, then
$(X,C,E)$ is a Markov chain with transition probability matrix $\Delta$,
defined by
$$
\Delta_{(x,c,e), (y,d,f)} = \Pi_{xy} \cdot q_d   \cdot p(f | y,d)
$$
and $\Phi, \Delta$ define a hidden Markov chain, denoted $Z(\Delta, \Phi)$.

We claim that  the output process $Z$, described in the first paragraph of this section, fits into this alternative
framework (when $X$ is a first order Markov chain). To see this,
let $\mathcal{E} = \mathcal{X} \times \mathcal{C} \times \mathcal{Z}$, and define
$p(e = (x,c,z)|x',c') = p(z|x,c)$ if $x = x'$ and $c = c'$, and $0$ otherwise.
Define $\Phi(x',c', (x,c,z)) = z$.  Then, $Z = Z(\Delta, \Phi)$ is a hidden Markov chain.
So, from hereon we adopt the alternative framework.

Now, we assume that $X$ is an irreducible first order Markov chain and that
the channel is parameterized by $\eps$ such that for
each $x, c,$ and $e$,  $p(e|x,c)(\eps)$  are analytic functions of $\eps
\ge 0.$ For each $\eps \ge 0$, let $\Delta(\eps)$ denote the corresponding transition probability
matrix on state set $\mathcal{X} \times \mathcal{C} \times \mathcal{E}$
and $\{Z(\eps)\}$ denote the family of resulting output hidden Markov chains. We also assume that
there is a one-to-one function from $\mathcal{X}$ into $\mathcal{Z}$, $z  = z(x)$, such that for all $c$, 
$p(z(x)|x,c)(0) = 1$. In other words, $\eps$ behaves like a ``composite index'' indicating how good the channel is, and small $\eps$ corresponds to the high signal-to-noise ratio. Then one can
verify that $\Delta(0)$ is a weak black hole and  $\Delta(\eps)$ is normally parameterized.
Thus, by Theorem~\ref{main}, we obtain an asymptotic formula for $H(Z(\eps))$ around $\eps =0$.
We remark that the above naturally generalizes to the case where
$X$ is a higher order irreducible Markov chain (through appropriately grouping matrices into blocks).

In the remainder of this section, we give three examples to illustrate the idea.

\begin{exmp}[Binary Markov Chains Corrupted by BSC($\eps$)]  \label{treated}

Consider a binary symmetric channel with crossover probability
$\varepsilon$. At time $n$ the channel can be characterized by the following equation
$$
Z_n=X_n \oplus E_n,
$$
where $\{X_n\}$ denotes the input process, $\oplus$ denotes binary addition, $\{E_n\}$ denotes the i.i.d.
binary noise with $p_E(0)=1-\varepsilon$ and $p_E(1)=\varepsilon$, and $\{Z_n\}$ denotes the corrupted output.
Note that this channel only has one channel state, and at $\eps=0$, $p_{Z|X}(1|1)=1, p_{Z|X}(0|0)=1$, so it fits in
the alternative framework described in the beginning of Section~\ref{III}.

Indeed, suppose $X$ is a first order irreducible Markov chain with the
transition probability matrix
$$
\Pi=\left[\begin{array}{cc}
       \pi_{00}&\pi_{01}\\
       \pi_{10}&\pi_{11}\\
       \end{array}\right].
$$
Then $Y=\{Y_n\}=\{(X_n, E_n)\}$ is jointly Markov with transition probability matrix (the column and row indices of the following matrix are ordered alphabetically):
$$
\Delta = \left [ \begin{array}{cccc}
\pi_{00} (1-\varepsilon) & \pi_{00} \varepsilon & \pi_{01} (1-\varepsilon) & \pi_{01} \varepsilon\\
\pi_{00} (1-\varepsilon) & \pi_{00} \varepsilon & \pi_{01} (1-\varepsilon) & \pi_{01} \varepsilon\\
\pi_{10} (1-\varepsilon) & \pi_{10} \varepsilon & \pi_{11} (1-\varepsilon) & \pi_{11} \varepsilon\\
\pi_{10} (1-\varepsilon) & \pi_{10} \varepsilon & \pi_{11} (1-\varepsilon) & \pi_{11} \varepsilon\\
\end{array} \right ],
$$
and $Z=\Phi(Y)$ is a hidden Markov chain with $\Phi(0, 0)=\Phi(1, 1)=0$, $\Phi(0,
1)=\Phi(1, 0)=1$. When $\eps =0$,
$$
\Delta = \left [ \begin{array}{cccc}
\pi_{00}  & 0 & \pi_{01}  & 0\\
\pi_{00} & 0 & \pi_{01}  & 0\\
\pi_{10}  & 0 & \pi_{11}  & 0\\
\pi_{10}  & 0 & \pi_{11}  & 0\\
\end{array} \right ], \Delta_0 = \left [ \begin{array}{cccc}
\pi_{00}  & 0 & 0  & 0\\
\pi_{00}  & 0 & 0  & 0\\
\pi_{10}  & 0 & 0  & 0\\
\pi_{10}  & 0 & 0  & 0\\
\end{array} \right ], \Delta_1 = \left [ \begin{array}{cccc}
0  & 0 & \pi_{01}  & 0\\
0  & 0 & \pi_{01}  & 0\\
0  & 0 & \pi_{11}  & 0\\
0  & 0 & \pi_{11}  & 0\\
\end{array} \right ],
$$
thus both $\Delta_0$ and $\Delta_1$ have rank one. If $\pi_{ij}$'s are all positive, then
we have a Black Hole case, for which one can derive the Taylor series
expansion of $H(Z)$ around $\eps=0$~\cite{zu05, hm06b}; if
$\pi_{00}$ or $\pi_{11}$ are zero, then this is a weak Black hole case with normal parameterization (of $\eps$), for which Theorem~\ref{main} can be applied and an asymptotic formula for $H(Z)$ around $\eps=0$ can be derived.

For a first order Markov chain $X$ with the
following transition probability matrix
$$
\left[ \begin{array}{cc}
1-p & p \\
1& 0
\end{array} \right],
$$
where $0\leq p \leq 1$, it has been shown~\cite{or04} that
$$
H(Z)=H(X)-\frac{p(2-p)}{1+p} \eps \log \eps +O(\eps)
$$
as $\eps \to 0$. This result has been further generalized~\cite{hm06a, jss07} to the following formula:
\begin{equation} \label{twoterms}
H(Z) = H(X) + f(X) \eps \log (1/\eps) + g(X) \eps + O(\eps^2 \log \eps),
\end{equation}
where $X$ is the input Markov chain of any order $m$ with transition probabilities $P(X_t=a_0|X_{t-m}^{t-1}=a_{-m}^{-1})$,
$a_{-m}^0 \in \mathcal{X}^m$, where $\mathcal{X}=\{0, 1\}$, $Z$ is the output process obtained by passing
$X$ through a BSC($\eps$), and $f(X)$ and $g(X)$ can be explicitly computed. Theorem~\ref{main} can be used to
generalize (\ref{twoterms}) to a formula with higher asymptotic terms. In particular, when $P(X_t=a_0|X_{t-m}^{t-1}=a_{-m}^{-1}) > 0$ for $a_{-m}^0 \in \mathcal{X}^{m+1}$, we have a Black Hole, in which
case, the Taylor series expansions of $H(Z)$ around $\eps=0$ can be explicitly computed (in principle); when $P(X_t=a_0|X_{t-m}^{t-1}=a_{-m}^{-1}) = 0$ for some $a_{-m}^0 \in \mathcal{X}^{m+1}$, we have a weak Black Hole, in which case an asymptotic formula of $H(Z)$ around $\eps=0$ can be obtained.

\end{exmp}

\begin{exmp}[Binary Markov Chains Corrupted by BEC($\eps$)] \label{IV}

Consider a binary erasure channel with fixed erasure rate $\eps$
(denoted by BEC($\eps$)). At time $n$ the channel can be characterized by the following equation
$$
Z_n=\left\{\begin{array}{cc}
             X_n & \mbox{ if } E_n=0\\
             e   & \mbox{ if } E_n=1\\
             \end{array}\right. ,
$$
where $\{X_n\}$ denotes the input process, $e$ denotes the erasure, $\{E_n\}$ denotes the i.i.d.
binary noise with $p_E(0)=1-\varepsilon$ and $p_E(1)=\varepsilon$,
and $\{Z_n\}$ denotes the corrupted output. Again this channel only has one channel state, and at $\eps=0$, $p_{Z|X}(1|1)=1, p_{Z|X}(0|0)=1$, so it fits in the alternative framework described in the beginning of Section~\ref{III}.

If the input $X$ is a first order irreducible Markov chain with transition probability
matrix
$$
\Pi=\left[\begin{array}{cc}
       \pi_{00}&\pi_{01}\\
       \pi_{10}&\pi_{11}\\
       \end{array}\right],
$$
and let $Z$ denote the output process. Then $Y=(X, E)$ is jointly
Markov with (the column and row indices of the following matrix are ordered alphabetically)
$$
\Delta = \left [ \begin{array}{cccc}
\pi_{00} (1-\eps) & \pi_{00} \eps & \pi_{01} (1-\eps) & \pi_{01} \eps\\
\pi_{00} (1-\eps) & \pi_{00} \eps & \pi_{01} (1-\eps) & \pi_{01} \eps\\
\pi_{10} (1-\eps) & \pi_{10} \eps & \pi_{11} (1-\eps) & \pi_{11} \eps\\
\pi_{10} (1-\eps) & \pi_{10} \eps & \pi_{11} (1-\eps) & \pi_{11} \eps\\
\end{array} \right ],
$$
and $Z=\Phi(Y)$ is hidden Markov with $\Phi(0, 1)=\Phi(1,
1)=e$, $\Phi(0, 0)=0$ and $\Phi(1, 0)=1$.

Now one checks that
$$
\Delta_0 = \left [ \begin{array}{cccc}
\pi_{00} (1-\eps) & 0 & 0 & 0\\
\pi_{00} (1-\eps) & 0 & 0 & 0\\
\pi_{10} (1-\eps) & 0 & 0 & 0\\
\pi_{10} (1-\eps) & 0 & 0 & 0\\
\end{array} \right ], \Delta_1 = \left [ \begin{array}{cccc}
0 & 0 & \pi_{01} (1-\eps) & 0\\
0 & 0 & \pi_{01} (1-\eps) & 0\\
0 & 0 & \pi_{11} (1-\eps) & 0\\
0 & 0 & \pi_{11} (1-\eps) & 0\\
\end{array} \right ], \Delta_e = \left [ \begin{array}{cccc}
0 & \pi_{00} \eps & 0 & \pi_{01} \eps\\
0 & \pi_{00} \eps & 0 & \pi_{01} \eps\\
0 & \pi_{10} \eps & 0 & \pi_{11} \eps\\
0 & \pi_{10} \eps & 0 & \pi_{11} \eps\\
\end{array} \right ].
$$
One checks that $\Delta(\eps)$ is normally parameterized by $\eps$ and thus Theorem~\ref{main} can be applied. Furthermore, Theorem~\ref{main} can be applied to the case when the input is an $m$-th order irreducible Markov chain $X$ to obtain asymptotic formula for $H(Z)$ around $\eps=0$.

\end{exmp}

\begin{exmp}[Binary Markov Chains Corrupted by Special Gilbert-Elliot Channel] \label{V}

Consider a binary Gilbert-Elliot channel, whose channel state (denoted by $C=\{C_n\}$) varies as an i.i.d. binary stochastic process with $p_C(0)=q_0, p_C(1)=q_1$ (here the channel
state varies as an i.i.d. process, rather than a generic Markov process). At time $n$ the channel can be characterized by the following equation
$$
Z_n=X_n \oplus E_n,
$$
where $\{X_n\}$ denotes the input process, $\oplus$ denotes binary addition, $\{E_n\}$ denotes the i.i.d.
binary noise with $p_{E|C}(0|0)=1-\eps_{0}$, $p_{E|C}(0|1)=1-\eps_{1}$, $p_{E|C}(1|0)=\eps_{0}$, $p_{E|C}(1|1)=\eps_{1}$
and $\{Z_n\}$ denotes the corrupted output. For such a channel, $p_{Z|(X, C)}(1|1, c)=1, p_{Z|(X, C)}(0|0, c)=1$ at $\eps=0$ for any channel state $c$. So it fits in the alternative framework described in the beginning of Section~\ref{III}.

To see this in more detail, we consider the special case when the input $X$ is a first order irreducible Markov chain with transition probability matrix
$$
\Pi=\left[\begin{array}{cc}
       \pi_{00}&\pi_{01}\\
       \pi_{10}&\pi_{11}\\
       \end{array}\right],
$$
and let $Z$ denote the output process. Then $Y=(X, C, E)$ is jointly
Markov with (the column and row indices of the following matrix are ordered alphabetically)
{\scriptsize
$$
\Delta = \left [ \begin{array}{cccccccc}
\pi_{00} q_{0} (1-\eps_0) & \pi_{00} q_{0} \eps_0 & \pi_{00} q_{1} (1-\eps_1) & \pi_{00} q_{1} \eps_1&\pi_{01} q_{0} (1-\eps_0) & \pi_{01} q_{0} \eps_0 & \pi_{01} q_{1} (1-\eps_1) & \pi_{01} q_{1} \eps_1 \\
\pi_{00} q_{0} (1-\eps_0) & \pi_{00} q_{0} \eps_0 & \pi_{00} q_{1} (1-\eps_1) & \pi_{00} q_{1} \eps_1&\pi_{01} q_{0} (1-\eps_0) & \pi_{01} q_{0} \eps_0 & \pi_{01} q_{1} (1-\eps_1) & \pi_{01} q_{1} \eps_1 \\
\pi_{00} q_{0} (1-\eps_0) & \pi_{00} q_{0} \eps_0 & \pi_{00} q_{1} (1-\eps_1) & \pi_{00} q_{1} \eps_1&\pi_{01} q_{0} (1-\eps_0) & \pi_{01} q_{0} \eps_0 & \pi_{01} q_{1} (1-\eps_1) & \pi_{01} q_{1} \eps_1 \\
\pi_{00} q_{0} (1-\eps_0) & \pi_{00} q_{0} \eps_0 & \pi_{00} q_{1} (1-\eps_1) & \pi_{00} q_{1} \eps_1&\pi_{01} q_{0} (1-\eps_0) & \pi_{01} q_{0} \eps_0 & \pi_{01} q_{1} (1-\eps_1) & \pi_{01} q_{1} \eps_1 \\
\pi_{10} q_{0} (1-\eps_0) & \pi_{10} q_{0} \eps_0 & \pi_{10} q_{1} (1-\eps_1) & \pi_{10} q_{1} \eps_1&\pi_{11} q_{0} (1-\eps_0) & \pi_{11} q_{0} \eps_0 & \pi_{11} q_{1} (1-\eps_1) & \pi_{11} q_{1} \eps_1 \\
\pi_{10} q_{0} (1-\eps_0) & \pi_{10} q_{0} \eps_0 & \pi_{10} q_{1} (1-\eps_1) & \pi_{10} q_{1} \eps_1&\pi_{11} q_{0} (1-\eps_0) & \pi_{11} q_{0} \eps_0 & \pi_{11} q_{1} (1-\eps_1) & \pi_{11} q_{1} \eps_1 \\
\pi_{10} q_{0} (1-\eps_0) & \pi_{10} q_{0} \eps_0 & \pi_{10} q_{1} (1-\eps_1) & \pi_{10} q_{1} \eps_1&\pi_{11} q_{0} (1-\eps_0) & \pi_{11} q_{0} \eps_0 & \pi_{11} q_{1} (1-\eps_1) & \pi_{11} q_{1} \eps_1 \\
\pi_{10} q_{0} (1-\eps_0) & \pi_{10} q_{0} \eps_0 & \pi_{10} q_{1} (1-\eps_1) & \pi_{10} q_{1} \eps_1&\pi_{11} q_{0} (1-\eps_0) & \pi_{11} q_{0} \eps_0 & \pi_{11} q_{1} (1-\eps_1) & \pi_{11} q_{1} \eps_1 \\
\end{array} \right ],
$$
}
and $Z=\Phi(X, C, E)$ is hidden Markov with
$$
\Phi(0, 0, 0)=\Phi(0, 1, 0)=\Phi(1, 0, 1)=\Phi(1, 1, 1)=0,
$$
$$
\Phi(0, 0, 1)=\Phi(0, 1, 1)=\Phi(1, 0, 0)=\Phi(1, 1, 0)=1.
$$

For some positive $k$, let $\eps_0=\eps, \eps_1= k \eps$. If $\eps=0$, one checks that
$$
\Delta_0 = \left [ \begin{array}{cccccccc}
\pi_{00} q_{0}  & 0 & \pi_{00} q_{1}  & 0&0  & 0 & 0  & 0 \\
\pi_{00} q_{0}  & 0 & \pi_{00} q_{1}  & 0&0  & 0 & 0  & 0 \\
\pi_{00} q_{0}  & 0 & \pi_{00} q_{1}  & 0&0  & 0 & 0  & 0 \\
\pi_{00} q_{0}  & 0 & \pi_{00} q_{1}  & 0&0  & 0 & 0  & 0 \\
\pi_{10} q_{0}  & 0 & \pi_{10} q_{1}  & 0&0  & 0 & 0  & 0 \\
\pi_{10} q_{0}  & 0 & \pi_{10} q_{1}  & 0&0  & 0 & 0  & 0 \\
\pi_{10} q_{0}  & 0 & \pi_{10} q_{1}  & 0&0  & 0 & 0  & 0 \\
\pi_{10} q_{0}  & 0 & \pi_{10} q_{1}  & 0&0  & 0 & 0  & 0 \\
\end{array} \right ], \Delta_1 = \left [ \begin{array}{cccccccc}
0  & 0 & 0  & 0&\pi_{01} q_{0}  & 0 & \pi_{01} q_{1}  & 0 \\
0  & 0 & 0  & 0&\pi_{01} q_{0}  & 0 & \pi_{01} q_{1}  & 0 \\
0  & 0 & 0  & 0&\pi_{01} q_{0}  & 0 & \pi_{01} q_{1}  & 0 \\
0  & 0 & 0  & 0&\pi_{01} q_{0}  & 0 & \pi_{01} q_{1}  & 0 \\
0  & 0 & 0  & 0&\pi_{11} q_{0}  & 0 & \pi_{11} q_{1}  & 0 \\
0  & 0 & 0  & 0&\pi_{11} q_{0}  & 0 & \pi_{11} q_{1}  & 0 \\
0  & 0 & 0  & 0&\pi_{11} q_{0}  & 0 & \pi_{11} q_{1}  & 0 \\
0  & 0 & 0  & 0&\pi_{11} q_{0}  & 0 & \pi_{11} q_{1}  & 0 \\
\end{array} \right ].
$$
So, both $\Delta_0$ and $\Delta_1$ will be rank one matrices and one can check that $\Delta(\eps)$ is normally parameterized by $\eps$. Again, Theorem~\ref{main} can be applied to the case when the input is an $m$-th order irreducible Markov chain $X$ to obtain an asymptotic formula for $H(Z)$ around $\eps=0$.

\end{exmp}

\end{document}